\documentclass[12pt]{article}
\pdfoutput=1
\usepackage[utf8]{inputenc}
\usepackage{amsmath}
\usepackage{amsthm}
\usepackage{amssymb}
\usepackage{float}
\usepackage[algoruled,linesnumbered,noend]{algorithm2e}
\usepackage{relsize}
\usepackage{enumitem}
\usepackage{cite}
\usepackage{multirow}
\usepackage{caption}
\usepackage{url}
\usepackage{xspace}
\usepackage{booktabs}

\usepackage{verbatim}
\newcommand{\nota}[1]{%
  {\sffamily\bfseries #1}%
  \marginpar{\framebox{\Large *}}%
}
\newcommand{\notaestesa}[2]{%
  \marginpar{\color{red!75!black}\textbf{\texttimes}}%
  {\color{red!75!black}%
    [\,\textbullet\,\textsf{\textbf{#1:}} %
    \textsf{\footnotesize#2}\,\textbullet\,]}%
}
\usepackage{tikz}
\usetikzlibrary{shapes,arrows}
\usetikzlibrary{positioning,patterns}
\usepackage{ifpdf}\usepackage{datetime}
\ifpdf%
\fi
\usepackage{graphicx}
\usepackage{a4wide}
\newcommand{\ie}{i.e.,~}

\newtheorem{theorem}{Theorem}
\newtheorem{lemma}[theorem]{Lemma}
\newtheorem{proposition}[theorem]{Proposition}

\newtheorem{definition}{Definition}

\begin{document}
\title{A New Lightweight Algorithm to compute the
  BWT and the LCP array of a Set of Strings}
\author{Paola Bonizzoni \and
  Gianluca Della Vedova \and
  Serena Nicosia \and
  Marco Previtali \and
  Raffaella Rizzi}

\date{DISCo, Univ. Milano-Bicocca, Milan, Italy}

\maketitle

\begin{abstract}
Indexing of very large collections of strings such as those produced by the widespread sequencing technologies, heavily relies on 
multi-string generalizations of the Burrows-Wheeler Transform (BWT), and
for this problem various in-memory algorithms have been proposed.    
The rapid growing of data that are processed routinely, such as in bioinformatics, requires a large amount of main memory, and this fact has motivated the development of algorithms, to compute the BWT, that work almost entirely in external memory.

On the other hand, the related problem of computing the  Longest
Common Prefix (LCP) array is often instrumental in several algorithms on collection of strings, such as those that compute the suffix-prefix
overlap among strings, which is an essential step for many genome
assembly algorithms.   

The best current lightweight approach to compute BWT and LCP array on a set of $m$ strings, each one $k$ characters long, has I/O complexity that is $O(mk^2 \log |\Sigma|)$ 
(where $|\Sigma|$ is the size of the alphabet), thus it is not optimal.    
        
In this paper we propose a novel approach to build BWT and LCP array (simultaneously) with $O(kmL(\log k +\log \sigma))$ I/O complexity,
where $L$ is the length of longest substring that appears at least twice in the input strings.
\end{abstract}

\section{Introduction}
In this paper we address the problem of costructing in external memory the Burrows-Wheeler Transform (BWT)
and the Longest Common Prefix (LCP) array for a large collection of strings. 
An efficient 
indexing of very large collections of strings is strongly
motivated by the widespread use of Next-Generation Sequencing (NGS) technologies that are
producing everyday collections of data that fill several terabytes of secondary storage,
that  has to be processed by sofware applications. 
Common applications in metagenomics require  indexing of  collections of strings (reads)
that are sampled from several genomes, where those genomes amount to billions of base
pairs.  For example, over 500 gigabases of data have been analyzed to start a
catalogue of the human gut microbiome~\cite{qin2010human}.


The Burrows-Wheeler Transform (BWT)~\cite{Burrows1994} is a reversible transformation of a
text that was originally designed for text compression; it is used for example in the
BZIP2 program.  
The BWT of a text $T$ is a permutation of its symbols and is strictly related to the
Suffix Array of $T$. 
In fact, the $i^{th}$ symbol of the BWT is the symbol preceding the $i^{th}$ smallest suffix of $T$ according to the lexicographical sorting of the suffixes of $T$.
The Burrows-Wheeler Transform has gained importance beyond its initial purpose, and has
become the basis for self-indexing structures such as the FM-index~\cite{Ferragina2005}, which
allows to efficiently perform important tasks such as searching a pattern in a text~\cite{Ferragina2005,Li15112014,Rosone2013}.
The generalization of the BWT (and the FM-index) to a collection of strings has been
introduced in~\cite{MantaciCPM2005,MantaciTCS2007}. 

An entire generation of recent bioinformatics tools heavily rely on the notion of BWT.    
For example, representing the reference genome with its FM-index is the basis of the most
widely used aligners, such as Bowtie~\cite{Langmead2009},
BWA~\cite{LiBWABioinformatics2009,LiBWABioinformatics2010} and
SOAP2~\cite{LiSOAP2Bioinformatics2009}.     

Still, to attack some other fundamental bioinformatics problems, such as genome assembly,
an all-against-all comparison among the input strings is needed, especially to
find all prefix-suffix matches (or overlaps) between reads in the context of the Overlap Layout Consensus (OLC) approach based on string graph~\cite{Myers2005}. This fact justifies
to search for extremely time and space efficient algorithms to compute the BWT on a collection of strings~\cite{li_exploring_2012,Valimaki2010,Ferragina2012,bauer_lightweight_2013}.    
For example, SGA (String Graph Assembler)~\cite{Simpson2012} is a de novo genome assembler that builds a string graph from the FM-index of the collection of input reads. 
In a preliminary version of SGA~\cite{Simpson2010}, the authors  estimated, for human sequencing data at a 20x coverage, the need of
700Gbytes of RAM in order to build the suffix array, using the construction algorithm
in~\cite{nong2009linear}, and the
FM-index.

Another technical device that is used to tackle the genome assembly in the OLC approach is the Longest Common
Prefix (LCP) array of a collection of strings, which is instrumental to compute (among others) the prefix-suffix matches in the collection.
The huge amount of available biological data has stimulated the development  of
the first efficient external-memory algorithms (called, BCR and BCRext) to construct the BWT of a collection of strings ~\cite{Bauer2011}.  
Similarly, a lightweight approach to the construction of the LCP array has been
investigated~\cite{DBLP:conf/wabi/BauerCRS12}. 
Towards an external memory genome assembler,
LSG~\cite{DBLP:conf/wabi/BonizzoniVPPR14,Bonizzoni2015} is founded upon BCRext and builds in external memory the string graph of a set of strings.    
In that approach, external memory algorithms to compute the BWT and the LCP
array~\cite{bauer_lightweight_2013,DBLP:conf/wabi/BauerCRS12} are fundamental.     

Still, the construction of the BWT (and LCP array) of a huge collection of strings is a
challenging task. 
A simple approach is constructing the BWT from the Suffix Array, but it is prohibitive for massive datasets. 
A first attempt to solve this problem~\cite{Siren2009}
partitions the input collection into batches, computes the BWT for each batch and then merges the results.

In this paper we present a new lightweight (external-memory) approach to compute the BWT and the LCP array of a collection of strings, which is alternative to BCRext~\cite{Bauer2011}.    
The algorithm BCRext is proposed together with BCR and both are designed to work on huge collections of strings (the experimental  analysis is on billions of 100-long strings).    
Those algorithms are lightweight because, on a collection of $m$ strings of length $k$,
BCR uses only $O(m \log(mk))$ RAM space and $O(km +  sort(m))$ CPU time, where $sort(m)$ is the
time taken to sort $m$ integers. The same complexity holds for the lightweight LCP algorithm given in~\cite{DBLP:conf/wabi/BauerCRS12}. 
Though the use of the RAM is negligible for DNA data,  the overall I/O volume is $O(k^2m + mk \log(mk))$. 
Clearly, a main question is if it is possible to achieve the optimal  $O(km)$ I/O complexity. 
Both BCR and BCRext build the BWT with a column-wise approach, where at each step $i$ the
elements preceding the suffixes of length $k - i - 1$ of each read are inserted in the
correct positions of the \emph{partial} BWT that considers only suffixes shorter than $k -
i - 1$.
Moreover, both algorithms are described as a succession of sequential scans, where the
partial BWTs are read from and and written to external files, thus obtaining a small main
memory footprint. 

Compared to BCRext, our algorithm uses an I/O volume that is  $O(L k m \log k)$, where $L$
is the maximum length of any substring appearing at least twice in the input collection.    
Clearly $L\le k$.    
Compared with BCR, our approach does not require an explicit sort of a generic set, but it
is mainly based on the simple idea of building partial BWTs, each one for the set of suffixes of a given length $l$, then merging those partial BWTs to obtain the complete BWT by using an
approach similar to the one proposed in~\cite{Holt2014}, where the construction of a
multi-string BWT is proposed with the main goal of merging BWTs for distinct genomic
sequences.

\section{Preliminaries}
\label{sec:definitions}

Let $\Sigma = \{c_0, c_1, \cdots, c_{\sigma}\}$ be a finite alphabet where $c_0 = \$$ (called \emph{sentinel}), and $c_0 < c_1 \cdots < c_{\sigma}$ where $<$ specifies the lexicographic ordering over alphabet $\Sigma$.  
We consider a collection  $S=\{s_1, s_2, \cdots, s_m\}$ of $m$ strings, where each string $s_j$ consists of $k$ symbols  over the
alphabet $\Sigma \setminus \{\$ \}$ and is terminated by the symbol \$.
%
The $i^{th}$ symbol of string $s_j$ is denoted by $s_j[i]$  and the substring $s_j[i]s_j[i+1] \cdots  s_j[t]$ of $s_j$ is denoted by $s_j[i:t]$. In order to simplify the presentation, we assume that all the strings in $S$ have the same length $k$.  
The \emph{suffix} and \emph{prefix} of $s_j$ of length $l$ are the
substrings $s_j[k-l +1: k]$ (denoted by $s_j[k-l +1:]$) and $s_j[1: l]$ (denoted by
$s_j[:l]$) respectively.
Then the \emph{$l$-suffix} and \emph{$l$-prefix} of a string $s_j$ is the suffix and prefix with length $l$, respectively.
The lexicographic ordering among strings in $S$ is defined in the usual way. Though we use the same sentinel to terminate strings, we can easily distinguish the same suffix of different strings by assuming an implicit ordering of the sentinels that is induced by the ordering of the input strings. More precisely, we assume that given $s_i, s_j \in S$, with $i < j$, then the sentinel of $s_i$ precedes the sentinel of $s_j$.


Given the lexicographic ordering $X$ of the suffixes of $S$, the \emph{Suffix Array} is the $(m(k+1))$-long array $SA$ where the element $SA[i]$ is equal to $(p, j)$ if and only if the $i^{th}$ element of $X$ is the  $p$-suffix of string $s_{j}$. 
The \emph{Burrows-Wheeler Transform (BWT)} of $S$ is the $(m(k+1))$-long array $B$ where if $SA[i] = (p,j)$, then $B[i]$ is the first symbol of the $(p+1)$-suffix of $s_j$ if $p<k$, otherwise $B[i]= \$$. 
In other words $B$ consists of the symbols preceding the ordered suffixes of $X$. 
%
%
The \emph{Longest Common Prefix (LCP) array} of  $S$ is the $(m(k+1))$-long array
$LCP$  such that $LCP[i]$  is the length of the longest prefix shared by suffixes $X[i-1]$ and $X[i]$.  Conventionally, $LCP[1]=-1$.

Now, we give the definition of \emph{interleave} of a generic set of arrays, that will be used extensively in the following.

\begin{definition}
\label{def:interleave}
Given $n+1$ arrays $V_0, V_1, \cdots, V_n$, then an array $W$ is an \emph{interleave} of $V_0, V_1, \cdots, V_n$ if
$W$ is the result of merging the arrays such that: (i) there is a 1-to-1 function $\psi_W$ from the set $\cup_{i=0}^n \{ (i,j): 1\le j\le |V_i|\}$ to the set $\{ q : 1\le q \le |W|\}$, (ii)
$V_i[j] = W[\psi_W(i,j)]$ for each $i,j$, and (iii) $\psi_W(i, j_1) < \psi_W(i, j_2)$ for each $j_1 < j_2$.
\end{definition}

By denoting with $L = \sum_{i=0}^{n} |V_i|$ the total length of the arrays, the interleave $W$ is a $L$-long array giving a fusion of $V_0, V_1 \cdots, V_n$ which preserves the relative order of the elements in each one of the arrays. As a consequence, for each $i$ with $0 \leq i \leq n$, the  $j^{th}$ element of $V_i$ corresponds to the  $j^{th}$ occurrence in $W$ of an element of $V_i$. This fact allows to encode the function $\psi_W$ as a $L$-long array $I_W$ such that $I_W[q] = i$ if and only if $W[q]$ is an element of $V_i$. Given $I_W$, it is possible to reconstruct $W$ by considering that $W[q]$ is equal to $V_{I_W[q]}[j]$ where $j$ is the number of values equal to $I_W[q]$ in the interval $I_W[1,q]$; this number will be called \emph{rank} at position $q$.
In the following, we will refer to vector $I_W$ as \emph{interleave-encoding} (or simply \emph{encoding}).
Algorithm~\ref{alg:interleave-encoding} shows how to reconstruct an interleave from its encoding (the array $rank$ is used to store the rank values), and can also be used to simulate a scan of $W$ by means of its encoding $I_W$.

\begin{algorithm}[htb!]
 
	\For{$i\gets 0$ to $n$}{%
		$rank[i] \gets 0$\;
	}

	\For{$q \gets 1$ to $|I_W|$}{%
		$i \gets I_W[q]$\;
		$rank[i] \gets rank[i] +1$\;
		$W[q] \gets V_i[rank[i]]$\;
	}
	
	\caption{Reconstruct the interleave $W$ from the encoding $I_W$}
	\label{alg:interleave-encoding}
\end{algorithm}

\section{The lightweight algorithm for BWT and LCP array}

Let $B_l$ and $X_l$ ($0 \le l \le k$) be $m$-long arrays such that $B_l[i]$ is the symbol preceding the $i^{th}$ smallest $l$-suffix of $S$ and $X_l[i]$ is the $i^{th}$ smallest $l$-suffix of $S$.
%
%
It is easy to see that the BWT $B$ is an interleave of the $k+1$ arrays $B_0, B_1, \cdots, B_k$, since the ordering of symbols in $B_{l}$ ($0 \leq l \leq k$) is preserved in $B$, i.e. $B$ is \emph{stable} w.r.t. each array $B_0, B_1, \cdots, B_k$. This fact is a direct consequence of the definition of $B$ and $B_{l}$. 
For the same reason, the lexicographic ordering $X$ of all suffixes of $S$ is an interleave of the arrays $X_0, X_1, \cdots, X_k$.
Let $I_B$ be the encoding of the interleave of arrays $B_0, B_1, \cdots, B_k$ giving the BWT $B$, and let $I_X$ be the encoding of the interleave of arrays $X_0, X_1, \cdots, X_k$ giving $X$. Then it is possible to show that $I_B = I_X$.

Our algorithm for building the BWT $B$ and the LCP array, differently from~\cite{Bauer2011}, consists of two distinct phases: in the first phase the arrays $B_0, B_1, \cdots , B_{k}$ are computed, while the second phase determines $I_X$ (which is equal to $I_B$) thus allowing to reconstruct $B$ as an interleave of $B_0, B_1, \cdots, B_k$. 
Indeed, BCRext~\cite{Bauer2011} computes the BWT of the collection  $S$  incrementally via $k+1$ iterations. 
At each iteration $l$, with $0\le l\le k$, the algorithm computes a partial BWT
$bwt_{l}(S)$ that is the BWT for the ordered collection of  suffixes of length at
most  $l$, that is for the lexicographic ordering of $X_0, X_1,  \cdots, X_{l}$.
This approach requires that, at each iteration $l$,  the symbols preceding the
$(l-1)$-suffixes of $S$ must be inserted  at their correct positions into $bwt_{l-1}(S)$,
that is each  $l$  iteration simulates the insertion of the $l$-suffixes in the ordered
collection of $\cup_{i=0}^{l-1}X_{i}$.  Updating  the partial BWT $bwt_{l}(S)$ in external memory, the
process  requires a sequential visit of the file containing the basic information of the
partial $bwt_{l -1}(S)$. Thus the I/O volume at each  iteration $l$  is at least $m (l-1) \log \sigma$ (since there are  $m$ suffixes for each length $i$ between $1$ to $l -1$). 
Consequently the total I/O volume for computing $bwt_{k}(S)$ is  at least $O(m k^2)$. 
More precisely, the BCRext algorithm in \cite{Bauer2011} that uses less RAM,  requires at each $l$ iteration an additional  I/O volume  given by $m \log (km)$, due to a process of ordering special arrays used to save RAM space. 
Our algorithm instead consists of a first phase that
has $O(m k)$ I/O volume  and time complexity  and produces the arrays $B_0, B_1, \cdots, B_k$ (see procedure {\bf Partition-suffixes}),
and a second phase which computes $I_X$ by
implicitly merging the arrays $X_{0}, X_1, \cdots , X_{k}$ into the interleave $X$ of the overall ordered set of all suffixes (see procedure {\bf Merge-suffixes}). As described in Section~\ref{sec:merge-suffixes}, the procedure does not need to compute explicitly the arrays $X_0, X_1, \cdots, X_k$ and the interleave $X$. Inspired by~\cite{Holt2014}, we perform this step by a number of $L$ iterations, where $L$ the length of the longest substring that has at least two occurrences in $S$. 
Thus the merging operation takes fewer iterations than BCRext (the latter requires $k$).

\section{The Procedure Partition-suffixes}

The input set $S=\{s_1, s_2, \cdots, s_m\}$ is preprocessed in order to have a fast access to its symbols, and $k$ $m$-long arrays   $S_0, S_1, \cdots, S_{k-1}$ are obtained. More in detail, the element $S_l[i]$ ($0 \leq l \leq k-1$)  is the $(k-l)^{th}$ symbol of the string $s_i$, that is $s_i[k-l]$.
In other words $S_l[i]$ is the symbol preceding the $l$-suffix of $s_i$.
The procedure Partition-suffixes (see Algorithm~\ref{alg:sort-suffixes}) takes in input  the arrays   $S_0, S_1, \cdots, S_{k-1}$ and computes the arrays $B_0, B_1, \cdots, B_k$
by using $k+1$ $m$-long arrays $N_l$ ($0 \leq l \leq k$), where $N_l[i]=q$ if and only if the $l$-suffix of the input string $s_q$ is the $i^{th}$ element of $X_l$.
%
Notice that 
the symbol $B_l[i]$ precedes the $l$-suffix $s_q[k-l+1:]$, that is $B_l[i] =s_q[k-l]$.
In particular, $N_0$ contains the sequence of indexes
$(1,2,3,\cdots, |S|)$ and $B_0$ contains the sequence $\langle s_1[k], s_2[k], \cdots, s_m[k] \rangle$ of the last symbols of the input strings (i.e. the symbols before the sentinels).

In order to specify the structure of the Procedure Partition-suffixes, given a symbol $c_h$ of the alphabet $\Sigma$,   we define the \emph{$c_h$-projection operation ${\Pi}_{c_h}$} over the array   $N_l$ that consists in taking from $N_l$  only the entries $i$ such that $s_i[k-l] =c_h$. In other words ${\Pi}_{c_h}(N_l)$ is the vector  that projects the entries of $N_l$ corresponding to strings whose $l$-suffix is preceded by the symbol $c_h$.
Then the following Lemma directly follows from definition of $N_{l  -1}$.

\begin{lemma}
Given the array $N_{l-1}$,  the  sequence of indexes of strings, whose $l$-suffix starts with symbol $c_h$ and ordered w.r.t. the $l$-suffix, is equal to vector ${\Pi}_{c_h}(N_{l  -1})$.
\end{lemma}

As a main consequence of the above Lemma the array $N_{l}$ can be simply obtained from $N_{l-1}$ as the concatenation  ${\Pi}_{{c}_0}(N_{l  -1}) \cdot {\Pi}_{{c}_1}(N_{l  -1}) \cdots  {\Pi}_{c_{\sigma}}(N_{l-1})$ where ${c}_0 \cdot {c}_1 \cdots {c}_{\sigma}$ is the lexicographic order of symbols of alphabet $\Sigma$.
Notice that  the ${c_h}$-projection of $N_{l  -1}$,  
is computed 
by listing the positions $i$ of $N_{l-1}$ such that $B_{l-1}[i]=c_h$. Indeed, $B_{l-1}$ lists the symbols precedings the ordered $(l-1)$-suffixes.  

The procedure Partition-suffixes computes arrays $B_0, \cdots, B_k$ in $k$ iterations.
At each iteration $l$, arrays $B_l$ and $N_l$ are computed from arrays $B_{l-1}$ and $N_{l-1}$.
The array $N_{l-1}$ is stored in $|\Sigma|$ lists $N_{l-1}(c_h)$, where  $N_{l-1}(c_h)$ is the $c_h$-projection of $N_{l-1}$. In the following, the arrays are treated as lists which can be stored in external files.

The basic procedure to compute $N_{l}$  from $B_{l-1}$ 
and $N_{l-1}$ is the following. First, $B_{l -1}$ is sequentially read and, for each position $i$, $N_{l-1}[i]$ is appended to the list $N_{l-1}(c_h)$, where $c_h=B_{l-1}[i]$. At this point,
$N_l$ is given by the concatenation of lists $N_{l-1}(c_0)N_{l-1}(c_1) \cdots N_{l-1}(c_\sigma)$.
After computing $N_{l}$, the vector $B_{l}$ can be obtained. Indeed, assuming that the $j^{th}$ element in the ordered list of $l$-suffixes is the suffix of string $i$ (that is, $N_{l}[j]=i$) the symbol preceding such suffix is $s_i[|s_i| - l]$ and is directly obtained by accessing position $i$ of vector $S_l$ (recall that $S_l$ has been computed in the preprocessing phase).
More precisely, $N_l$ is sequentially read and, for each position $j$, if $N_{l}[j] = i$ then $B_{l}[j] = S_l[i]$.
Due to a random access, array $S_l$ 
it is assumed to be kept in RAM with a space cost of $O(m \log \sigma)$. 

\begin{algorithm}[htb!]
	\SetKwInOut{Input}{Input}\SetKwInOut{Output}{Output}
	\Input{The arrays $S_0, \cdots, S_k$.}
	\Output{The arrays $B_0, \cdots, B_k$.}
	\For{$i\gets 1$ to $m$}{%
		$B_0[i] \gets S_0[i]$\;
        $N_{0}[i] \gets i$\;
	}
	\For{$l\gets 1$ to $k$}{%
		\ForEach{$c \in\Sigma$}{%
			$N_{l-1}(c) \gets $ empty list\;
		}
        $B_l \gets$ empty list\;
        $N_l \gets$ empty list\;
		\For{$i\gets 1$ to $m$}{%
	        $c \gets B_{l-1}[i]$\;
			Append $N_{l-1}[i]$ to  $N_{l-1}(c)$\;
        }
		\For{$i\gets 0$ to $|\Sigma|$}{%
			Append $N_{l-1}(c_i)$ to $N_l$\;
		}
		\For{$i\gets 1$ to $m$}{%
			Append $S_{l}[N_{l}[i]]$ to $B_{l}$\;		
        }
	}
	\caption{Partition-suffixes}
	\label{alg:sort-suffixes}
\end{algorithm}

\section{The procedure \textbf{Merge-suffixes}}
\label{sec:merge-suffixes}

The second step of our algorithm computes the encoding $I_X$ of the interleave $X$ of the arrays $X_{0}, X_1, \cdots , X_{k}$, giving the lexicographic ordering of all suffixes of $S$ and (at the same time) computes the LCP array. Recall that $I_X$ is equal to the encoding $I_B$ of the interleave of the arrays $B_{0}, B_1, \cdots , B_{k}$ giving the BWT $B$.
This section is devoted to describe how to compute $I_{X}$ from which it is easy to obtain the BWT $B$ as explained in Algorithm~\ref{alg:interleave-encoding}, while the description of the approach to obtain the LCP array is postponed until Section~\ref{sec:computing-lcp-array}.

Before entering into the details, we need some definitions.

\begin{definition}
\label{def:prec-p-relation}
Let $\alpha = s_{i_{\alpha}}[k-l_{\alpha}+1:]$ and $\beta = s_{i_{\beta}}[k-l_{\beta}+1:]$ be
two generic suffixes of $S$, with length respectively $l_{\alpha}$ and $l_{\beta}$.    
Then, given an integer $p$, $\alpha \prec_p \beta$ (and we say that $\alpha$ \emph{p-precedes} $\beta$) iff one of the
following conditions hold: (1) $\alpha[:p]$ is lexicographically strictly smaller than
$\beta[:p]$, (2) $\alpha[:p]=\beta[:p]$ and $l_{\alpha} < l_{\beta}$, (3)
$\alpha[:p]=\beta[:p]$, $l_{\alpha} = l_{\beta}$ and $i_{\alpha} < i_{\beta}$. 
\end{definition}

\begin{definition}
\label{definition:p-interleave-X}
Given the arrays $X_0, X_1, \cdots, X_k$, the \emph{$p$-interleave} $X^p$ ($0 \leq p \leq k$) is the interleave such that $X^p[i]$ is the $i^{th}$ smallest suffix in the $\prec_p$-ordering of all the suffixes of $S$.
\end{definition}

It is immediate to verify that $X^k$ (that is, the suffixes sorted according to the $\prec_{k}$ relation) is equal to $X$, hence $I_X=I_{X^k}$.  
Therefore, our approach is to determine $I_{X^k}$ by iteratively computing
$I_{X^p}$ by increasing values of $p$, starting from $I_{X^0}$. Observe that $X^0$ lists the suffixes in the same order given by the concatenation of arrays $X_0, X_1, \cdots, X_k$ and the encoding $I_{X^0}$ is trivially given by $|X_0|$ $0$s, followed by $|X_1|$ $1$s, \ldots, followed by $|X_k|$ values equal to $k$.

\begin{definition}
\label{def:p-segment}
Let $X^{p}$ be the $p$-interleave of $X_0, X_1, \cdots, X_k$, and let $i$ be a position.    
Then, the \emph{$p$-segment} of $i$ in $X^{p}$ is the maximal interval $[b,e]$ such that  $b\le i\le e$ and all suffixes in $X^p[b,e]$ have the same $p$-prefix.    
Positions $b$ and $e$ are called respectively \emph{begin} and \emph{end}
position of the segment, and the common $p$-prefix is denoted by $w_p(b,e)$.
\end{definition}

It is immediate to observe that the set of all the $p$-segments of  a $p$-interleave form a partition
of its positions $(1,\cdots, (k+1)m)$.
Observe that, by definition, a suffix smaller than $p$ belongs to a $p$-segment $[b,e]$ having $b=e$. In other words, such suffix is the unique element of the $p$-segment.

Before describing the approach, the computation of $X^p$ from $X^{p-1}$ is explained. Let $Q^p_l$ ($0 \le p \le k$ and $0 \le l \le k$) be the $m$-long array such that $Q^p_l[i]$ is the $p^{th}$ symbol of the suffix $X_l[i]$. In particular, $Q_{l}^{p}[i]$ is the sentinel \$ if the suffix is smaller than $p$.
Moreover, let 
$Q^p$ be the interleave of  the arrays $Q^p_0, Q^p_1, \cdots, Q^p_k$ such that $I_{Q^p} = I_{X^{p-1}}$. In other words, $Q^p[i]$ is the $p^{th}$ symbol of the suffix $X^{p-1}[i]$. 

\begin{lemma}  
\label{lemma:pi-function}
Let $[b,e]$ be a $(p-1)$-segment of $X^{p-1}$. Then, $X^p[b,e]$ is a permutation of $X^{p-1}[b,e]$ defined by the permutation $\Pi^{p-1}_{b,e}$ of the indexes $(b, b+1, \cdots, e)$ producing the stable ordering of the symbols in $Q^p[b,e]$, such that the $r^{th}$ suffix of $X^p[b,e]$ is the suffix of $X^{p-1}$ in position $\Pi^{p-1}_{b,e}[r]$.
\end{lemma}

\begin{proof}
First we prove that $X^p[b,e]$ is a permutation of $X^{p-1}[b,e]$. Let us denote with $w$ the $(p-1)$-prefix common to suffixes in $X^{p-1}[b,e]$, and let $i$ be a position in $[b,e]$.
Given a position $q < b$, by definition, the $(p-1)$-prefix $w_q$ of $X^{p-1}[q]$ is strictly smaller than $w$. Then, the $p$-prefix of $X^{p-1}[q]$ is strictly smaller than the $p$-prefix of $X^{p-1}[i]$.
In the same way, given a position $q' > e$ by definition, the $(p-1)$-prefix $w_q'$ of $X^{p-1}[q']$ is strictly greater than $w$. Then, the $p$-prefix of $X^{p-1}[q']$ is strictly greater than the $p$-prefix of $X^{p-1}[i]$. Hence, the set of the suffixes of $X^{p-1}$ before $b$ and the set of the suffixes after $e$ are equal
(respectively) to the set of the suffixes of $X^p$ before $b$ and to the set of the suffixes after $e$, thus deriving that for $b \le i \le e$ the suffix $X^{p-1}[i]$ is equal to $X^p[j]$ for some $j$ in $[b,e]$, completing the proof of the first part.

Furthermore, all suffixes in $X^{p-1}[b,e]$ share the common $(p-1)$-prefix $w$, and therefore their $\prec_p$-order can be determined by ordering their $p^{th}$ symbols. More specifically, the suffix $X^{p-1}[i]$ ($b \le i \le e$) is the $r^{th}$ suffix in $X^p[b,e]$, where $r$ is the rank of its $p^{th}$ character in the stable order of $Q^p[b,e]$.
\end{proof}

Given the suffix in position $i$ of $X^{p-1}$, such that $i$ is in the $(p-1)$-segment $[b,e]$, the Lemma \ref{lemma:pi-function} allows to compute its position $i' \in [b,e]$ on $X^p$. Let $\#^{<}$
be the number of symbols of $Q^p[b,e]$ that are strictly smaller than $Q^p[i]$ and let $\#^{=}_{q}$
be the number of symbols of $Q^p[b,q]$ which are equal to $Q^p[i]$. Then, the rank of suffix $X^{p-1}[i]$ in $X^p[b,e]$ is $r=\#^{<} + \#^{=}_{i}$,
thus deriving that its position in $X^p$ is $i' = b + r -1$. 
It is possible to notice that the positions $(b, b+1, \cdots, e)$ on $X^p$ are partioned into $n$ $p$-segments $[b,e_1], \cdots, [b_n, e]$ (referred as \emph{induced} by the $(p-1)$-segment $[b,e]$ of $X^{p-1}$), where $n$ is the number of distinct non-\$ symbols in $Q^p[b,e]$ plus the number $\#_\$$ of symbols \$ in $Q^p[b,e]$.
Observe that the first $\#_\$$ $p$-segments $[b_1,e_1], \cdots, [b_{\#_\$}, e_{\#_\$}]$ have width $1$, while the width of the last $n-\#_\$$ $p$-segments $[b_{\#_\$+1}, e_{\#_\$+1}], \cdots, [b_n, e_n]$ can be computed as follows.
Let $\{c_1, \cdots, c_{n-\#_\$}\}$ be the ordered set of the distinct non-\$ symbols in $Q^p[b,e]$. Then, the width of $[b_{\#_\$+i}, e_{\#_\$+i}]$ ($1 \leq i \le n-\#_\$$) is equal to the number of occurrences of the symbol $c_i$ in $Q^p[b,e]$.
From what described above, it derives that the $p$-segments on $X^p$ form a partition of its positions $(1, \cdots , (k+1)m)$ that is a refinement of the partition formed by the $(p-1)$-segments on $X^{p-1}$.

Now, we describe a simple procedure (see Algorithm~\ref{alg:merge-compute-precp-order}) to compute $X^p[b,e]$ from the $(p-1)$-segment $[b,e]$ of $X^{p-1}$. The procedure uses $\sigma+1$ (initially empty) lists $L_{c_0}, \cdots, L_{c_{\sigma}}$.  
Each position $i \in [b,e]$ is considered from $b$ to $e$ and each suffix $X^{p-1}[i]$ is appended to the list $L_c$ such that $c$ is the $p^{th}$ symbol of $X^{p-1}[i]$.
Afterwards, each list of the sequence $\langle L_{c_0}, \cdots, L_{c_{\sigma}} \rangle$ is read sequentially, and each suffix $x$ in position $r_h$ of the list $L_{c_h}$ is put in the position $b + r -1$ of $X^p$, where $r$ is given by the total size of the lists $L_{c_0}, \cdots, L_{c_{h-1}}$ (which have been previously read) plus the position $r_h$ of $x$. Observe that $r$ contains the rank of $x$ in the $\prec_p$-ordering of the suffixes of the $(p-1)$-segment.

\begin{algorithm}[htb!]
$L_{c_0}, L_{c_1}, \cdots, L_{c_{\sigma}} \gets $ empty lists\;

\For{$i \gets b$ to $e$}{%
  $c \gets Q^p[i]$\;
  Append $X^{p-1}[i]$ to $L_c$\;
}

$r \gets 1$\;

\For{$h \gets 0$ to $\sigma$}{%
  \For{$r_h \gets 1$ to $|L_{c_h}|$}{%
    $x \gets L_{c_h}[r_h]$\;			
    $X^p[b + r - 1] \gets x$\;
    $r \gets r + 1$\;
  }
}
\caption{Compute $p$-segment on $X$ from a $(p-1)$-segment}
\label{alg:merge-compute-precp-order}
\end{algorithm}

Algorithm~\ref{alg:merge-compute-precp-order} can be easily modified in order to produce
also the $p$-segments $[b, e_1], [b_2, e_2], \cdots, [b_n ,e]$ 
induced by the $(p-1)$-segment $[b,e]$.
Observe that the first $|L_{c_0}|$ $p$-segments have width $1$, while for the last
$n-|L_{c_0}|$ $p$-segments it is easy to prove that $e_i = b + T_i -1$ where $T_i$ is the
total size of the first $i-|L_{c_0}|+1$ nonempty lists $L_{c_0}, L_{c_1}, \cdots,
L_{c_{\sigma}}$. 
The entire interleave $X^p$ is obtained by computing $X^p[b,e]$ for each distinct $(p-1)$-segment $[b,e]$ of $X^{p-1}$. 

At this point, it is immediate to extend the definition of $p$-segment $[b,e]$ from $X^{p}$ to its encoding $I_{X^{p}}$, and to see that the Algorithm~\ref{alg:merge-compute-precp-order} can be slightly modified to compute $I_{X^p}[b,e]$ from the $(p-1)$-segment $[b,e]$ of $I_{X^{p-1}}$ (see Algorithm~\ref{alg:merge-compute-precp-order-on-encoding}).

\begin{algorithm}[htb!]
	
	$L_{c_0}, L_{c_1}, \cdots, L_{c_{\sigma}} \gets $ empty lists\;
	
	\For{$i \gets b$ to $e$}{%
		$c \gets Q^p[i]$\;
		Append $I_{X^{p-1}}[i]$ to $L_c$\;
	}
	
	$r \gets 1$\;
	
	\For{$h \gets 0$ to $\sigma$}{%
		\For{$r_h \gets 1$ to $|L_{c_h}|$}{%
			$j \gets L_{c_h}[r_h]$\;			
			$I_{X^p}[b + r - 1] \gets j$\;
			$r \gets r + 1$\;
		}
	}
	
        \caption{Compute $p$-segment on $I$ from a $(p-1)$-segment}
	\label{alg:merge-compute-precp-order-on-encoding}
\end{algorithm}

Based on Algorithm~\ref{alg:merge-compute-precp-order-on-encoding} we designed the iterative procedure \textbf{Merge-suffixes} (see Algorithm~\ref{alg:merge-suffixes}) to compute the encoding $I_{X^k}$ starting from the encoding $I_{X^0}$ that can be easily obtained as explained before.
Recall that $I_{X^k}$ is the encoding of the interleave of the arrays $B_0, \cdots, B_k$ giving the BWT $B$ of the input set $S$.
The iteration $p$ of the procedure computes $I_{X^{p}}$ from  $I_{X^{p-1}}$, 
by scanning the array $I_{X^{p-1}}$, and is detailed in Algorithm~\ref{alg:merge-compute-iteration}. Precisely, the procedure, for each $(p-1)$-segment $[b,e]$, computes the portion $I_{X^p}[b,e]$ of $I_{X^p}$. We point out that it is not actually necessary to reconstruct the interleave $Q^p$ from the arrays $Q^p_0, Q^p_1, \cdots, Q^p_k$, since its encoding is $I_{X^{p-1}}$, and therefore a scan of $I_{X^{p-1}}$ allows also to simulate a scan of $Q^p$ (see Algorithm~\ref{alg:interleave-encoding}). 

\begin{algorithm}[htb!]
	
	$L_{c_0}, L_{c_1}, \cdots, L_{c_{\sigma}} \gets $ empty lists\;
	\For{$j\gets 0$ to $k$}{%
		$rank[j] \gets 0$\;
	 }
	
	$pick\_up\_start \gets true$\;
	
	\For{$i\gets 1$ to $(k + 1)m$}{%
		\If{pick\_up\_start = true}{
			$b \gets i$\;
			$pick\_up\_start = false$\;
		}
		$j \gets I_{X^{p-1}}[i]$\;
		$rank[j] \gets rank[j] + 1$\;
		$c \gets Q_{index}^{p}[rank[j]]$\;
		Append $j$ to $L_c$\;
        \If{$i$ is the end position of a $(p-1)$-segment}{\label{alg:merge-compute-iteration:pick-up-start}
			$pick\_up\_start = true$\;
			$r \gets 1$\;
				
			\For{$h \gets 0$ to $\sigma$}{%
				\For{$r_h \gets 1$ to $|L_{c_h}|$}{%
					$j \gets L_{c_h}[r_h]$\;			
					$I_{X^p}[b + r - 1] \gets j$\;
					\If{$r_h > 1$ and $h >0$}{
						$Lcp_p[b+r-1]=p$\; \label{alg:merge-compute-iteration:update-lcp}
					}\Else{
						$Lcp_p[b+r-1]=Lcp_{p-1}[b+r-1]$\;
					}
					$r \gets r + 1$\;
				}
			}
			$L_{c_0}, L_{c_1}, \cdots, L_{c_{\sigma}} \gets $ empty lists\;
		}
	}
		
	\caption{Compute $I_{X^{p}}$ from $I_{X^{p-1}}$}
	\label{alg:merge-compute-iteration}
\end{algorithm}

The conditions at line~\ref{alg:merge-compute-iteration:pick-up-start} of Algorithm~\ref{alg:merge-compute-iteration} and at line~\ref{alg:merge:outer-loop-begin} of Algorithm~\ref{alg:merge-suffixes} are checked by using an auxilary binary array $E^{p-1}$ storing the $(p-1)$-segments. More specifically, $E^{p-1}[i]$ is \emph{true} iff $i$ is the \emph{end} position of some $(p-1)$-segment. The array $E^{p-1}$ is sufficient to reconstruct the set of all $(p-1)$-segments since they form a partition of positions $(1, \cdots, (k+1)m)$, and it is read sequentially with the other arrays. For the sake of brevity the computation of $E^{p-1}$ (of each iteration $p$) is omitted.

Observe that, under the assumption that the input set $S$ does not contain duplicates,  all the $k$-segments of the encoding $I_{X^k}$ have width equal to $1$.
Moreover, after $L$ iterations, where $L$ is the length of the longest common substring of two strings in $S$, (1) the encoding $I_{X^L}$ is equal to $I_{X^k}$ and (2) each $I_{X^j}$ with $j > L$ is identical to $I_{X^L}$. 
Those two facts are a consequence of the following two observations:  (i) the length $p$ of the longest common prefix between two strings is equal to the length of the longest common substring in $S$, if all the $(p+1)$-prefixes of the suffixes are distinct, (ii) the  $\prec_{p+1}$ order relation does not effect the ordering given by $I_{X^p}$, that is $I_{X^{p+1}}=I_{X^{p}}$.

Algorithm~\ref{alg:merge-compute-precp-order-on-encoding} computes also the LCP array whose description is in the following Section~\ref{sec:computing-lcp-array}.
Section~\ref{sec:parte-su-q_p} is devoted to describe how to compute the arrays $Q_{l}^{p}$ used by iteration $p$.

\begin{algorithm}[htb!]
\SetKwInOut{Input}{Input}\SetKwInOut{Output}{Output}
\Input{The arrays $B_0, B_1, \cdots, B_k$}
\Output{The encoding $I_{X^k}$.}
\For{$l \gets 0$ to $k$\label{alg:merge:preprocessing-start}}{%
  \For{$i \gets 1$ to $m$}{%
    $I_{X^0}[l m+i] \gets l$;
    $Lcp[lm+i] \gets 0$\;
  }
}
Compute lists $Q^{1}_{l}$ for $0 \leq l \leq k$\;
$p\gets 1$\;
\While{there exists some $(p-1)$-segment on $I_{X^{p-1}}$ which is wider than $1$ \label{alg:merge:outer-loop-begin}}{%
  Compute $I_{X^p}$ from $I_{X^{p-1}}$\;
  Compute lists $Q^{p+1}_{l}$ for $0 \leq l \leq k$\;
}
Output $I_{X^p}$\;
\caption{Merge-suffixes}
\label{alg:merge-suffixes}
\end{algorithm}

\subsection{Computing the LCP array}
\label{sec:computing-lcp-array}


The LCP array is obtained by exploiting
Proposition~\ref{proposition:lcp-with-segments} which easily follows from the definition of $p$-segment.

\begin{proposition}
	\label{proposition:lcp-with-segments} 
	Let $i$ be a position on the LCP array $LCP$.    
	Then $LCP[i]$ is the largest $p$ such that $i$ is the start of a $(p+1)$-segment (of $I_{X^{p+1}}$) and is not the start of a $p$-segment (of $I_{X^p}$).
\end{proposition}

\begin{proof}
Notice that, since the $(p+1)$-segments are a refinement of the
$p$-segments, then there can be only one such $p$.    
Let  $\alpha_{i-1}$ and $\alpha_{i}$ be respectively the $(i-1)^{th}$
and the $i^{th}$ lexicographically smallest suffix of $S$.    
Assume initially that $i$ is the start of  a $(p+1)$-segment, but not
of a $p$-segment.
Since $i$ is not a start of a $p$-segment, then $i-1$ and $i$ belong
to the same $p$-segment hence, by definition of segment, they share
the same $p$-prefix.    
Since $i$ is the start of a $(p+1)$-segment, then $i-1$ and $i$ cannot belong
to the same $(p+1)$-segment, hence they do not share
the same $(p+1)$-prefix.    
Thus, $LCP[i] = p$.    
Assume now that $LCP[i] = p$, that is $\alpha_{i-1}$ and $\alpha_{i}$
share a common $p$-prefix, but not a $(p+1)$-prefix.    
Again, by definition of segment, $i-1$ and $i$ belong to the same
$p$-segment but not to the same $(p+1)$-segment.    
\end{proof}

At this point, let $Lcp_p$ be the $(k+1)m$-long array such that $Lcp_p[i]$ is the length of the longest common prefix between the $p$-prefix of suffix $X^p[i]$ 
and the $p$-prefix of suffix $X^p[i-1]$. The array $Lcp_k$ is clearly equal to the LCP array of the input set $S$.

Each iteration $p$ (see Algorithm~\ref{alg:merge-compute-precp-order-on-encoding}) of our procedure computes $Lcp_p$ from $Lcp_{p-1}$ and the array $Lcp_0$ is set to all $0$s before starting the iterations. The following invariant, which directly implies its correctness, is maintained.

\begin{lemma}  
\label{lemma:lcp-invariant}
At the end of iteration $p$, $Lcp_p[i]=p$ iff $i$ is not the start position of any $p$-segment.    
\end{lemma}

\begin{proof}
We will prove the lemma by induction on $p$.
Before the first iteration, the array $Lcp_0$ is set to all $0$s, therefore we only have to consider the general case.
%
Observe that (at the beginning of iteration $p$), given a $(p-1)$-segment $[b,e]$, we have $Lcp_{p-1}[i]=p-1$ for $b+1 \le i \le e$. Then, the procedure (see Algorithm~\ref{alg:merge-compute-precp-order-on-encoding}) sets to $p$ the array $Lcp_p$ in all positions of the induced $p$-segments different from their start positions (line~\ref{alg:merge-compute-iteration:update-lcp}), completing the proof.
\end{proof}

\subsection{Computing the $Q_{l}^{p}$ arrays}
\label{sec:parte-su-q_p}



In this section we describe how to compute the arrays $Q^p_0, Q^p_1, \cdots, Q^p_k$ used by iteration $p$. Recall that $Q^p_l$ is the $m$-long array such that $Q^p_l[i]$ is the
$p^{th}$ symbol of the $i^{th}$ smallest $l$-suffix of $X_l$
($Q_{l}^{p}[i]$ is a sentinel \$ if the suffix is smaller than $p$). 
The following proposition establishes a recursive definition of $Q_{l}^{p}$.

\begin{lemma}
\label{lemma:ordering-suffix}
Let $X_{l}$ and $X_{l -1}$ be respectively the sorted
$l$-suffixes and $(l -1)$-suffixes of the set $S$. 
Let $\alpha_{l}$ and $\alpha_{l -1}$ be respectively the
$l$-suffix and the $(l-1)$-suffix of a generic input string $s_{i}$.    
Then the $p^{th}$ symbol of  $\alpha_{l}$ is the $(p-1)^{th}$ symbol of  $\alpha_{l -1}$.    
\end{lemma}

Since the suffixes  $\alpha_{l}$ and $\alpha_{l -1}$ can have
different positions in $X_{l}$ and
$X_{l-1}$, the list $Q_{l}^{p}$ is a permutation of $Q_{l -1}^{p-1}$.    
Still, Algorithm~\ref{alg:shift-succs} exploits the construction of $Q_{l -1}^{p-1}$ to
quickly compute  $Q_{l}^{p}$.
Notice that, for $l \geq 1$, $Q^1_l$ is the result of sorting $B_{l-1}$ whereas for $l=0$,
$Q^1_0$ is a sequence of sentinels 
Therefore the arrays $Q^1_0, Q^1_2, \cdots, Q^1_k$ can be trivially computed.

\begin{algorithm}[htb!]
	\SetKwInOut{Input}{Input}\SetKwInOut{Output}{Output}
	\Input{The lists $B_{0}, \cdots, B_{k}$ on alphabet $c_0, \cdots , c_\sigma$, an integer $p$ with $2\le p\le k$, and all $Q_l^{p-1}$.}
	\Output{The lists $Q_{l}^{p}$ for each $k\ge l\ge p$}
	\For{$l\gets p$ to $k$\label{alg:shift-succs:begin-init}}{%
    	$Q_l^{p} \gets$ empty list\;
	  	\For{$h\gets 0$ to $\sigma$}{%
            $Q_l^{p}(c_h) \gets$ empty list\label{alg:shift-succs:end-init}\;
        }
    	\For{$j\gets 1$ to $m$}{%
         	Append $Q_{l-1}^{p-1}[j]$ to $Q_l^{p}(B_{l-1}[j])$\;
      	}
	  	\For{$h\gets 0$ to $\sigma$}{%
        	Append $Q_l^p(c_h)$ to $Q^{p}_{l}$\label{alg:shift-succs:end-shift}\;      
		}
	}
	\caption{Compute all lists $Q_l^p$ for any given $p\ge 2$.}
	\label{alg:shift-succs}
\end{algorithm}



In order to prove the correctness of Algorithm~\ref{alg:shift-succs} we need to show that the  permutation $St^{l-1}$ over indexes $1, \cdots, m$ of  $B_{l -1}$  induced by the lexicographic ordering of $B_{l -1}$, is the correct permutation  of $Q_{l -1}^{p-1}$ to obtain   $Q_{l}^{p}$.
Indeed, observe that $St^{l -1}$ is the permutation that relates positions of indexes of strings in $X_{l-1}$ to their positions in $X_l$. More precisely, given a string $s_q$ of $S$, such that its $(l-1)$-suffix is in position $j$ of list $X_{l-1}$, then if $St^{l-1}[j] =t$, it means that the $l$-suffix is of the string $s_q$ is in position $t$ of list $X_l$.

The above observation is a consequence of the fact that in order to get the lexicographic ordering of $X_l$ from the list $X_{l-1}$ we simply sort the $(l-1)$-suffixes by the first symbol that precedes them, i.e., they are sorted by the list $B_{l -1}$.

\section{Complexity}





First of all, notice that 
all the arrays are accessed sequentially,
therefore they can be stored in external files and it is immediate to view our procedure as an external memory approach, where only
the arrays $S_l$ ($0 \leq l \leq k-1$) of the symbols of the input strings are kept in main memory (together with some additional data structures).

First we will consider the procedure Partition-suffixes which mainly consists of $k$ iterations.
At iteration $l$, the coordinated scans of $N_{l-1}$ and $B_{l-1}$ suffice to construct
$N_l$ and successively $B_l$.
Notice that we keep array $S_{l}$ in main memory, since it is the only array that we access randomly.
This procedure has $O(km)$ I/O complexity and a matching running time.
Moreover, keeping the array $S_{l}$ in main memory requires $O(m\log \sigma)$ space.

Mainly the procedure Merge-suffixes is a loop where each iteration $p$ consists of a coordinated scan of the arrays $I_{X^{p-1}}$, $E^{p-1}$ and $Lcp_{p-1}$, and of the arrays $Q_l^p$ for $0 \leq l \leq k$, as well as writing the arrays $I_{X^p}$, $E^{p}$ and $Lcp_p$, and 
computing all the $Q_l^{p+1}$ arrays.
Both $Lcp_{p-1}$ and $I_{X^{p-1}}$ have $O(km)$ elements, each requiring $O(\log k)$ space, therefore their scan
implies $O(km\log k)$ I/O complexity.
The $E^p$ array has $O(km)$ elements, each requiring $O(\log \sigma)$ space, which implies an 
$O(km\log \sigma)$ I/O complexity.
There are at most $k$ $Q_l^p$ arrays, each consisting of $m$ elements where each element requires $O(\log \sigma)$ space, which implies an $O(km\log \sigma)$ I/O complexity.
There are some additional data structures, whose I/O complexity is smaller than the $O(km(\log k +\log \sigma))$ of the other parts.

The only relevant data structure of Merge-suffixes that must reside in main memory is the $rank$ array, which has $k$ elements, each requiring $O(\log k)$ space.

The last component of our complexity analysis is the number of iterations of Merge-suffixes.
Notice that the condition of the while loop at line~\ref{alg:merge:outer-loop-begin} is 
equivalent to testing whether all $p$-segments contain only one suffix.
Notice that, if $L$ is the length of the longest substring appearing at least twice, then each $(L+1)$-long substring appears once in the input strings $S$, that is
all $(L+1)$-prefixes of some suffixes are unique.    
Consequently, the procedure iterates exactly $L$ times over the loop at line~\ref{alg:merge:outer-loop-begin}.
Therefore, the overall I/O complexity of the algorithm is $O(kmL(\log k +\log \sigma))$

\section{Conclusions}

We have presented a new lightweight algorithm to compute the BWT and the LCP array of a set of strings, whose I/O complexity is competitive with BCRext.
More precisely, our overall I/O complexity is $O(kmL(\log k +\log \sigma))$, while BCRext 
has $O(mk (k\log \sigma + \log(mk)))$.

While our focus has been on the theoretical aspects, it would be interesting to implement the proposed algorithm and perform an experimental analysis to determine the practical behavior.
Since the number of iterations of the Merge-suffixes procedure is not fixed a priori, we
expect a finely tuned implementation to provide great improvements.

\newpage\section*{Time complexity}

{\bf BCR} requires $O(k^{2}\mathtt{sort}(m))$ time, $O(m\log(\sigma) + m\log(mk) +
m\log(m))$ main memory, and $O(mk^2\log(\sigma))$ I/O.
At each iteration $i$, it maintains in main memory 3 lists of length $m$ that contain the
characters that have to be inserted at the current step ($m\log(\sigma)$ bits), the
position where the characters have to be inserted ($m\log(mk)$ bits), and the indexes of
the reads sorted by their $k - i$ suffix  ($m\log(m)$ bits) respectively. Note that the
computation of the position where the current characters have to be inserted (\ie the rank
of such character) is performed by a sequential scan of the partial BWT and that at each
position an occurrence counter
is modified. This means that at each iteration {\bf BCR} performs $O(km)$ computations.

{\bf BCRext} requires $O(k^{2}m)$ time, $O(\sigma^2\log(mk))$ main memory, and $O(mk^2\log(\sigma) + mk\log(mk) + mk\log(m))$ I/O. It aims to lower the main memory requirement of {\bf BCR} by storing the 3 lists of the previous approach in external memory and accesses them sequentially. Moreover, this method requires to read and write the input sequences at each iteration and has an additional $O(mk^2\log(\sigma)$ I/O. At each iteration $i$, {\bf BCRext} sequentially reads the external files and implicitly sorts the reads by their $k - i$ suffix, thus obtaining the correct sorting of the elements to be added at iteration $i+1$.

\begin{table}
\begin{tabular}{p{3em}p{11em}p{11em}p{11em}}

~                & \texttt{BCR} (with LCP)              & \texttt{BCRext} (BWT only)           & This paper (with LCP)\\ \hline
CPU Time         & $O(k(m + \mathtt{sort}(m)))$         & $O(km)$                              & $O(kml)$ \\ 
~                & ~                                    & ~                                    & ~ \\
RAM usage (bits) & $O((m + \sigma^2)\log(mk))$ 	        & $O(\sigma^2\log(mk))$                & $O(m\log(\sigma))$ \\
~                & ~                                    & ~                                    & ~ \\
I/O (bits)       & $O(mk^2\log(\sigma))$ (partial BWT)  & $O(mk^2\log(\sigma))$ (partial BWT)  & $O(mk\log\sigma))$ (vectors $B_*$) \\
~                & $O(mk^2\log(k))$ (partial LCP)       & $O(mk^2\log(\sigma))$ (sequences)    & $O(mk\log(m))$ (vectors $N_*$) \\
~                & $O(mk\log(\sigma))$ (current symbols)& $O(mk\log(mk))$ (P array)            & $O(mkL\log(k))$ (interleave)\\
~                & ~                                    & $O(mk\log(m))$ (N array)             & $O(mkL\log(\sigma))$ (vectors $Q_*$) \\
~                & ~                                    & ~                                    & $O(mkL)$ (vectors $End_p$) \\

\end{tabular}
\caption{Time and space complexity comparison with \texttt{BCR} and \texttt{BCRext}.}
\end{table}

\bibliographystyle{plain}
\bibliography{bwt}

\end{document}